\documentclass[11pt]{article}

\usepackage{graphicx}
\usepackage{psfrag}
\usepackage{color}

\usepackage{latexsym}
\usepackage{amssymb}
\usepackage{gensymb}

\usepackage{amsmath}

\usepackage{simplemargins}
\setallmargins{1in}

\usepackage{enumerate,enumitem}

\newtheorem{theorem}{Theorem}
\newtheorem{lemma}{Lemma}[section]

\newenvironment{proof}{{\bf Proof:} }{\hspace*{\fill}$\Box$\vspace{2mm}}

\newtheorem{corollary}{Corollary}[section]

\newtheorem{proposition}{Proposition}[section]
\newtheorem{definition}{Definition}[section]

\title{\LARGE\bf Non-crossing Connectors in the Plane}
\author{\large Jan Kratochv\'{\i}l \hspace{5.1cm} Torsten~Ueckerdt\\[+1mm]
 {\tt honza@kam.mff.cuni.cz} \hspace{3.2cm} {\tt ueckerdt@googlemail.com} \hspace{0.8cm}\\[+1mm]
      \sf    Charles University in Prague,\\
      \sf    Prague, Czech Republic
}

\begin{document}

\maketitle

\begin{abstract}
 We consider the non-crossing connectors problem, which is stated as follows: Given $n$ simply connected regions $R_1,\ldots,R_n$ in the plane and finite point sets $P_i \subset R_i$ for $i=1,\ldots,n$, are there non-crossing connectors $\gamma_i$ for $(R_i,P_i)$, i.e., arc-connected sets $\gamma_i$ with $P_i \subset \gamma_i \subset R_i$ for every $i=1,\ldots,n$, such that $\gamma_i \cap \gamma_j = \emptyset$ for all $i \neq j$?

 We prove that non-crossing connectors do always exist if the regions form a collection of pseudo-disks, i.e., the boundaries of every pair of regions intersect at most twice. We provide a simple polynomial-time algorithm if the regions are axis-aligned rectangles. Finally we prove that the general problem is NP-complete, even if the regions are convex, the boundaries of every pair of regions intersect at most four times and $P_i$ consists of only two points on the boundary of $R_i$ for $i=1,\ldots,n$.
\end{abstract}


\section{Introduction}
Connecting points in the plane in a non-crossing way is a natural problem in computational geometry. For example, it is well known that one can always find a non-crossing matching with straight line segments between any set of $n$ red and $n$ blue points in the plane. Related problems ask for a non-crossing path that alternates between red and blue vertices, or two non-crossing spanning trees, one on each color, which minimize the number of crossings between them. We refer to the survey article of Kaneko and Kano~\cite{KanKan03} for more on red and blue points. Recent investigations consider the problem of finding a non-crossing matching between $n$ ordered pairs of points sets, e.g., between a single point and a vertical line~\cite{many10,many11}, or between two vertical line segments~\cite{Ver08}.

In this paper, we investigate what happens if we allow general curves instead of just straight line segments. Moreover, we want to connect not only $n$ pairs of points, but $n$ finite sets of points. Of course, we can always find such non-crossing curves, unless two point sets intersect. But if we impose for every point set a region that the corresponding curve must be contained in, then determining whether or not such non-crossing curves exist becomes a non-trivial problem.

\medskip
\noindent
\textbf{Related work.} A lot of research has been done for connecting points in the plane with straight line segments in a non-crossing way. For instance, a point set $P$ is \emph{universal} for a class $\mathcal{G}$ of planar graphs if the vertices of every graph $G \in \mathcal{G}$ can be embedded onto the points $P$ such that straight edges do not cross. It has been shown~\cite{PacGriMohPol91,CasUrr96} that every set of $n$ points in general position is universal for the class of $n$-vertex outer-planar graphs. The smallest point set that is universal for the class $\mathcal{G}_n$ of all $n$-vertex planar graphs consists of at least $1.235n$~\cite{ChrKar89,Kur04} and at most $\mathcal{O}(n^2)$~\cite{FraPacPol90,Sch90} points. Deciding whether a given graph embeds on a given point set, is known to be NP-complete~\cite{Cab04}.

In many variants edges are allowed to be more flexible than straight line segments. Kaufmann and Wiese~\cite{KauWie02} show that every $n$-element point set is universal for $\mathcal{G}_n$ if every edge is a polyline with at most $2$ bends. Moreover, some $n$-element point sets are universal for $\mathcal{G}_n$ if edges bend at most once~\cite{EveLazLioWis10}. If the bending points have to be embedded onto $P$ as well, then universal sets of size $\mathcal{O}(n^2 / \log n)$ for $1$ bend, $\mathcal{O}(n \log n)$ for $2$ bends, and $\mathcal{O}(n)$ for $3$ bends are known~\cite{DujEvaLazLenLioRapWis11}. A variant with so-called ortho-geodesic edges was studied especially for trees of maximum degree $3$ and $4$ ~\cite{DiGFraFulGriKru11}.

Variants where every vertex $v$ has an associated subset $P_v$ of $P$ of its possible positions has been studied for much simpler graphs like matchings and cycles, however for straight edges only. For cycles, deciding whether a set of non-crossing edges exists is known to be NP-complete, even if every $P_v$ is a vertical line segment or every $P_v$ is a disk~\cite{Lof07}. For matchings, NP-completeness has been shown if $|P_v| \leq 3$~\cite{many10,many11}, or $P_v$ is a vertical line segment of unit length~\cite{Ver08}. The latter result still holds if every edge $\{u,v\}$ may be a monotone curve within the convex hull of $P_u \cup P_v$~\cite{Spe11}.

\medskip
\noindent
\textbf{Our results.} In this paper we allow edges to be general curves, which w.l.o.g. can be thought of as polylines of finite complexity, that is with a finite number of bends. Given an $N$-element point set $P$, each curve $\gamma_i$ is asked to go through (connect) a fixed subset $P_i$ of $P$ of two or more points. Any two such curves, called \emph{connectors}, shall be non-crossing, i.e., have empty intersection. In particular, w.l.o.g. $P$ is partitioned into subsets $P_1,\ldots,P_n$ and we ask for a set of $n$ non-crossing connectors $\gamma_1,\ldots,\gamma_n$ with $P_i \subset \gamma_i$ for $i=1,\ldots,n$. It is easily seen that the order in which $\gamma_i$ visits the points $P_i$ may be fixed arbitrarily. Indeed, we could ask to embed any planar graph $G_i$ with $|P_i|$ vertices and curved edges onto $P_i$, even while prescribing the position of every vertex in $G_i$.

Non-crossing connectors as described above do always exist. But the situation gets non-trivial if we fix subsets $R_1,\ldots,R_n$ in the plane, called \emph{regions}, and impose $\gamma_i \subset R_i$ for every $i=1,\ldots,n$. Throughout this paper every $R_i$ is a simply connected region in the plane, which contains $P_i$.

In Section~\ref{sec:pseudo-disks} we prove that non-crossing connectors do always exist if the given regions form a collection of pseudo-disks, i.e., the boundaries of every pair of regions intersect at most twice. In Sections~\ref{sec:rectangles} and~\ref{sec:NP-complete} we consider the computational complexity of deciding whether or not non-crossing connectors exist. In particular, in Section~\ref{sec:rectangles} we show that the problem is polynomial if the regions are axis-aligned rectangles, while in Section~\ref{sec:NP-complete} we prove that the problem is NP-complete, even if the regions are convex, the boundaries of every pair of regions intersect at most four times, and $|P_i| = 2$ for every $i=1,\ldots,n$. We start with some notation in Section~\ref{sec:problem} and show that the non-crossing connectors problem is in NP.

\section{The non-crossing connectors problem}\label{sec:problem}
The non-crossing connectors problem is formally defined as follows.

\medskip
\noindent
\textsc{Non-crossing Connectors}

\noindent
\textbf{Given:} \hspace{5pt}\hangindent = 15pt Collection $R_1,\ldots,R_n$ of simply connected subsets of the plane and a finite point set $P_i \subset R_i$, for $i=1,\ldots,n$ with $P_i \cap P_j = \emptyset$ for $i\neq j$.

\noindent
\textbf{Question:} \hspace{5pt}\hangindent = 15pt Is there a collection $\gamma_1,\ldots,\gamma_n$ of curves, such that $P_i \subset \gamma_i \subset R_i$ for $i=1,\ldots,n$ and $\gamma_i \cap \gamma_j = \emptyset$ for $i\neq j$?

\medskip
\noindent
The boundary of every region $R_i$ is a simple closed curve, denoted by $\partial R_i$. We assume here and for the rest of the paper that $\partial R_i \cap \partial R_j$ is a finite point set. We may think of $\bigcup_{i=1}^n \partial R_i$ as an embedded planar graph $G = (V,E)$ with vertex set $V = \{p \in \mathbb{R}^2 \;|\; p \in \partial R_i \cap \partial R_j, i \neq j\}$ and edge set $E = \{e \subset \mathbb{R}^2 \;|\; e \text{ is a connected component of } \bigcup \partial R_i \setminus V\}$. A point $p \in \partial R_i \cap \partial R_j$ is either a \emph{crossing point} or a \emph{touching point}, depending on whether the cyclic order of edges in $\partial R_i$ and $\partial R_j$ around $p$ is alternating or not. We say that two regions $R_i, R_j$ are \emph{$k$-intersecting} for $k \geq 0$ if $|\partial R_i \cap \partial R_j| \leq k$ and all these points are crossing points, i.e., w.l.o.g. $k$ is even. A set $R_1,\ldots,R_n$ of regions is \emph{$k$-intersecting} if this is the case for any two of them. For example, $R_1,\ldots,R_n$ are $0$-intersecting if and only if\footnote{To be precise, we always require for the ``if''-part that $\partial R_i \cap \partial R_j$ is a finite set of crossing points for $i \neq j$.} they form a nesting family, i.e., $R_i \cap R_j \in \{\emptyset, R_i, R_j\}$ for $i \neq j$. 

Regions $R_1,\ldots,R_n$ are a called a \emph{collection of pseudo-disks} if and only if they are $2$-intersecting. Usually two pseudo-disks may have one touching point. However, this can be locally modified into two crossing points with affecting the existence of non-crossing connectors. Pseudo-disks for example include homothetic copies of a fixed convex point set, but they are not convex in general. A collection of axis-aligned rectangles is always $4$-intersecting, but not necessarily $2$-intersecting. Finally, if $R_1,\ldots,R_n$ are convex polygons with at most $k$ corners, then they are $2k$-intersecting.

We close this section by showing that the non-crossing connectors problem belongs to NP.

\begin{proposition}\label{prop:in-NP}
 \textsc{Non-crossing Connectors} is in NP.
\end{proposition}
\begin{proof}
 We reduce our problem to Weak Realizability of Abstract Topological Graphs. Abstract topological graphs (AT-graphs, for short) have been introduced in~\cite{KraLubNes91} as triples $(V,E,R)$ where $(V,E)$ is a graph and $R$ is a set of pairs of edges of $E$. The AT-graph $(V,E,R)$ is weakly realizable if $(V,E)$ has a drawing (not necessarily non-crossing) in the plane such that $ef\in R$ whenever the edges $e,f$ cross in the drawing. Weak realizability of AT-graphs is NP-complete. The NP-hardness was shown in~\cite{Kra91}, and the NP-membership was shown relatively recently in~\cite{SchSedSte03}.

 We assume the input of the problem be described as a plane graph $G$ (the boundaries of the regions) with the incidence structure of the points of $P$ to the faces of the graph. The size of the input is measured by the number of crossing points of the boundaries plus the number of regions and points to be connected. Strictly speaking $G$ may contain nodeless loops corresponding to the boundaries of regions which are not crossed by any other boundary. Such a case will be handled quickly by the next construction. We refer to Figure~\ref{fig:NP} for an illustrative example.

 Add vertices and edges to $G$ to create a vertex $3$-connected supergraph $G'$ of $G$. This can be achieved for instance by subdividing every edge of $G$ by at most $3$ new extra vertices (subdivide nodeless loops by $3$ vertices, loops by $2$ vertices and simple edges by a single vertex each so that to avoid multiple edges in $G'$), adding a new vertex in each face adjacent to all vertices of this face (both the original and the new ones) and adding edges connecting the subdividing vertices to create a triangulation. Such a graph $G'$ is $3$-connected and hence it has a topologically unique non-crossing drawing in the plane (up to the choice of the outerface and its orientation).

 \begin{figure}[htb]
  \centering
  \includegraphics{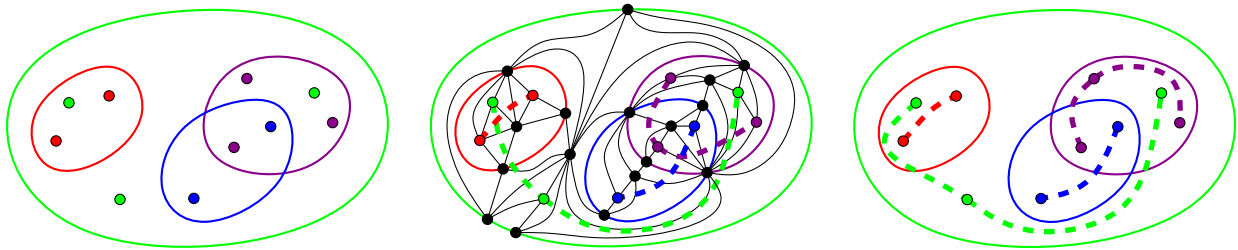}
  \caption{From Non-crossing Connectors to Weak Realizability of AT-graphs.}
  \label{fig:NP}
 \end{figure}

 For every point of $P$, choose a triangular face of $G'$ in the face of $G$ it lies in, and regard it as a vertex adjacent to the vertices of this triangular face (if more points of $P$ are assigned to the same face of $G'$, we add edges between them and to the vertices of the face to triangulate it). Call the resulting planar triangulation $G''$. As the last step, for every $i$, add edges that connect the points of $P_i$ by a path, and call the edges of these paths the {\em connecting edges}. Call the resulting graph $\widetilde{G}$.

 Now define an AT-graph with underlying graph $\widetilde{G}$ by allowing the connecting edges of points $P_i$ to cross anything but other connecting edges and the edges resulting from the boundary of $R_i$. No other edges are allowed to cross, in particular, no edges of $G''$ may  cross each other. It is straightforward to see that a weak realization  of this AT-graph is a collection of non-crossing connectors, and vice versa. Thus the NP-membership follows by the result of Schaefer et al.~\cite{SchSedSte03} and the fact that our construction of $\widetilde{G}$ is polynomial (if $G$ has $c$ crossing points, $a$ arcs, and $f$ faces,  and $p$ is the total number of points in $P$, then $\widetilde{G}$ has at most $c+3a+f+p$ vertices and at most $3c+9a+3f+4p$ edges).
\end{proof}

\section{Pseudo-Disks}\label{sec:pseudo-disks}

In this section we prove that non-crossing connectors do always exist if the given regions form a collection of pseudo-disks. We begin with an auxiliary lemma.

\begin{lemma}\label{lem:pseudo-disks}
 Let $R,R'$ be pseudo-disks, $p \in R \backslash R'$ be a point, and $\gamma \subset R \cap R'$ a curve that intersects $\partial R$ exactly twice. Then the connected component of $R \backslash \gamma$ not containing $p$ is completely contained in the interior of $R'$.
\end{lemma}
\begin{proof}
 Let $C$ be the connected component of $R \setminus \gamma$ not containing $p$. Let $q,r$ be the intersections of $\gamma$ with $\partial R$. Then $\partial R \setminus \{p,q,r\}$ is a set of three disjoint curves. The curve between $p$ and $q$ as well as between $p$ and $r$ contains a point in $\partial R'$, since $p \notin R'$ and $q,r \in \gamma \subset R'$. Because $R,R'$ are pseudo-disks $q$ and $r$ are the only points in $\partial R \cap \partial R'$ and hence the third curve $\delta = C \cap \partial R$ between $q$ and $r$ is completely contained in $R'$. Since the closed curve $\delta \cup \gamma \subset R'$ is the boundary of $C$ and $\delta \cap \partial R' = \emptyset$, we conclude that $C$ is completely contained in the interior of $R'$.
\end{proof}

\begin{theorem}\label{thm:pseudo-disks}
 If $R_1,\ldots,R_n$ is a collection of pseudo-disks, then non-crossing connectors exist for any finite point sets $P_i \subset R_i$ ($i= 1,\ldots,n$) with $P_i \cap P_j = \emptyset$ for $i \neq j$.
\end{theorem}
\begin{proof}
 The proof is constructive. Let $R_1,\ldots,R_n$ be a collection of pseudo-disks and $P_i$ a finite subset of $R_i$ for $i=1,\ldots,n$. We assume w.l.o.g. that every $\partial R_i$ is a closed polygonal curve of finite complexity. Moreover, we assume that the regions are labeled from $1,\ldots,n$, such that for every $i = 2,\ldots,n$ the set $R_i \backslash (R_1 \cup \cdots \cup R_{i-1})$ is non-empty, i.e., contains some point $p_i$. For example, we may order the regions by non-decreasing $x$-coordinate of their rightmost point. Note that rightmost points of pseudo-disks do not coincide. For simplicity we add $p_i$ to $P_i$ for every $i=1,\ldots,n$ (and denote the resulting point set again by $P_i$). Clearly, every collection of non-crossing connectors for the new point sets is good for the original point sets, too.

 We start by defining a connector $\gamma_1$ for $(R_1,P_1)$ arbitrarily, such that $P_1 \subset \gamma_1 \subset R_1$, $\gamma_1 \cap P_i = \emptyset$ for every $i \geq 2$, and $\gamma_1$ is a polyline of finite complexity. To keep the number of operations in the upcoming construction finite we concider polylines of finite complexity only. That is, whenever we define a curve we mean a polyline of finite complexity even if we do not explicitly say so.

 For $i=2,\ldots,n$ assume that we have non-crossing connectors $\gamma_1,\ldots,\gamma_{i-1}$, such that $(\bigcup_{j<i} \gamma_j) \cap (\bigcup_{k \geq i} P_k) =\emptyset$. We want to define a connector $\gamma_i$ for $(R_i,P_i)$. The set $R_i \setminus (\bigcup_{j<i} \gamma_j)$ has finitely many connected components $\{C_k\}_{k\in K}$ with $|K| < \infty$. Every point in $P_i$ is contained in exactly one $C_k$. Let $C_0$ be the component containing the additional point $p_i \in P_i$. The informal idea is the following. We reroute some of the existing connectors until $P_i$ is completely contained in $C_0$. Then, we define a connector $\gamma_i$ for $(R_i,P_i)$ arbitrarily, such that $P_i \subset \gamma_i \subset C_0$, as well as $\gamma_i \cap P_j = \emptyset$ for $j > i$ and $\gamma_i \cap \gamma_j = \emptyset$ for $j<i$. The reader may consider Figure~\ref{fig:pseudo-disks} for an illustration of the upcoming operation. For better readability the parts of connectors are omitted in Figure~\ref{fig:pseudo-disks}, which have an endpoint in the interior of $R_i$. However, those, as well as the point sets $P_j$ with $j > i$, will be circumnavigated by the curve $\delta$.

 \begin{figure}[htb]
  \centering
  \includegraphics{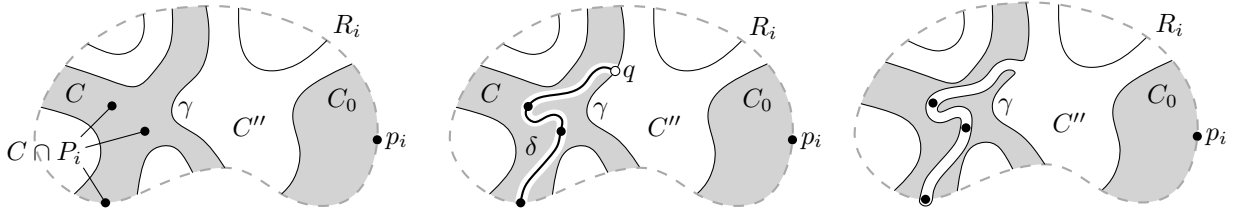}
  \caption{Rerouting the curve $\gamma$ bordering the components $C$ and $C''$, such that the subset of $P_i$ formerly contained in $C$ is contained in $C''$ afterwards.}
  \label{fig:pseudo-disks}
 \end{figure}

 Every connector $\gamma_j$ for $j<i$ is a simple curve. Thus every connected component of $R_i \setminus (\bigcup_{j<i} \gamma_j)$ contains a point from the boundary of $R_i$, i.e., the adjacency graph between the components is a tree $T$ on vertex set $\{C_k\}_{k\in K}$, which we consider to be rooted at $C_0$. Let $C \neq C_0$ be a component, such that $C\cap P_i \neq \emptyset$ but $C'\cap P_i = \emptyset$ for every descendant $C'$ of $C$ in $T$. Let $\gamma$ be the curve in $R_i$ that forms the border between $C$ and its father $C''$ in $T$, i.e., $\gamma$ intersects $\partial R_i$ only at its endpoints and is a subset of some connector $\gamma_{j^*}$ for $(R_{j^*},P_{j^*})$. In particular, $j^* < i$ and hence $p_i \notin R_{j^*}$. Applying Lemma~\ref{lem:pseudo-disks} with $p = p_i$ we get that $C$ is contained in the interior of $R_{j^*}$.

 Let $q \notin P_{j^*}$ be any interior point of $\gamma$ and $\delta$ be any curve with endpoint $q$, such that $(P_i \cap C) \subset \delta \subset (C \cup \{q\}) \subset R_{j^*}$, as well as $\delta \cap (\bigcup_{j<i} \gamma_j) = \{q\}$ and $\delta \cap (\bigcup_{j\neq i} P_j) = \emptyset$. We reroute $\gamma_{j^*}$ within a small distance around $\delta$. More formally, define a simply connected set $D \supset \delta$ to be a thickening of the curve $\delta$ by some small $\varepsilon >0$, such that $D$ is still contained in $R_{j^*} \setminus (\bigcup_{j \neq i} P_j \cup \bigcup_{j<i} \gamma_j)$. Note that $D \nsubseteq C$ if $P_i$ (and hence $\delta$) contains points on the boundary of $R_i$. However, we can ensure that $D \subset R_{j^*}$ since $C$ lies in the \emph{interior} of $R_{j^*}$. Moreover, we can choose $\varepsilon$ small enough, such that $\partial D$ intersects $\gamma$ only in two points $q_1$ and $q_2$, which are $\varepsilon$-close to $q$.

 Next, the part of $\gamma$ between $q_1$ and $q_2$ is replaced by the part of $\partial D$ between $q_1$ and $q_2$ that runs through $C$. This rerouting of $\gamma$ (and implicitly the connector $\gamma_{j^*}$) may (or may not) change the subtree of $T$ rooted at $C''$. But it does not affect any component of $R_i \setminus \bigcup_{j<i} \gamma_j$ that is not in this subtree. Moreover, $C''$ is extended by $D \cap C$, which contains all points in $P_i \cap C$. Hence, the so-to-speak total distance of the points in $P_i$ from $C_0$ in $T$ is decreased. After finitely many steps we have $P_i \subset C_0$ and thus can define the connector $\gamma_i$ for $(R_i,P_i)$.
\end{proof}

\section{Axis-Aligned Rectangles}\label{sec:rectangles}

Throughout this section the regions $R_1,\ldots,R_n$ are given as axis-aligned rectangles. Whenever we consider some axis-aligned rectangle $R$ we consider it as a closed set, i.e., $R = [x_1,x_2] \times [y_1,y_2]$ for some $x_1 < x_2$ and $y_1 < y_2$. In particular, we always take the closure without explicitly saying so.

\begin{definition}
 Two rectangles $R_i,R_j$ form a \emph{cross} (and $R_i$ and $R_j$ are called \emph{crossing}) if they are $4$-intersecting, and form a \emph{filled cross} if additionally both connected components of $R_i \backslash R_j$ contain a point from $P_i$, and both connected components of $R_j \backslash R_i$ contain a point from $R_j$.
\end{definition}

Note that if $R_i$ and $R_j$ are crossing, then every connected component $C$ of $(R_i \cup R_j) \setminus (R_i \cap R_j)$ is an axis-aligned rectangle, and hence we consider its closure in the above definition. Obviously, non-crossing connectors do not exist if some pair of rectangles is a filled cross. In other words, the absence of filled crosses is a necessary condition for the existence of non-crossing connectors. Next, we show that this condition is also sufficient.

\begin{theorem}\label{thm:rectangles}
 A set of axis-aligned rectangles admits a set of non-crossing connectors if and only if it does not contain a filled cross.
\end{theorem}
\begin{proof}
 The ``only if''-part is immediate. We prove the ``if''-part by applying Theorem~\ref{thm:pseudo-disks}. To this end we consider axis-aligned rectangles $R_1,\ldots,R_n$ such that no two of them form a filled cross. If there is no cross at all, then the rectangles are pseudo-disks and non-crossing connectors exist by Theorem~\ref{thm:pseudo-disks}. So assume that some pair of rectangles is a cross, but \emph{not} a filled cross. Consider such a cross $\{R_i = [x^i_1,x^i_2] \times [y^i_1,y^i_2], R_j = [x^j_1,x^j_2] \times [y^j_1,y^j_2]\}$ where $R_i \cap R_j$ is inclusion-minimal among all crosses. W.l.o.g. assume that $R_i \cap R_j = [x^j_1,x^j_2] \times [y^i_1,y^i_2]$ and the connected component $C = [x^j_1,x^j_2] \times [y^i_2,y^j_2]$ of $R_j \setminus R_i$ contains no point from $P_j$. The situation is illustrated in Figure~\ref{fig:rectangles}. Figuratively speaking we chop off $C$ (actually a slight superset $C'$ of $C$) from $R_j$ in order to reduce the total number of crosses of all rectangles. More precisely, choose $\varepsilon > 0$ small enough that $C' := [x^j_1,x^j_2] \times [y^i_2 - \varepsilon, y^i_2]$ contains no point from $P_j$, and that the $y$-coordinate of no corner of a rectangle $\neq R_i$ lies between $y^i_2 -\varepsilon$ and $y^i_2$. We replace $R_j$ by $\tilde{R_j} := R_j \setminus C'$.

 \begin{figure}[htb]
  \centering
  \includegraphics{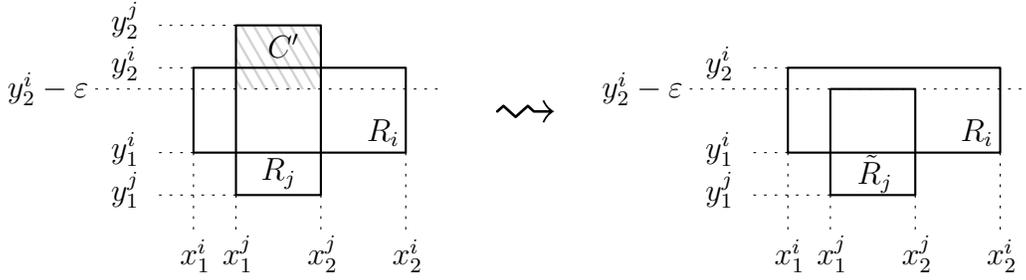}
  \caption{Chopping off $C'$ with $C' \cap P_j = \emptyset$ from the rectangle $R_j$ to obtain $\tilde{R}_j$.}
  \label{fig:rectangles}
 \end{figure}

 We claim that every rectangle $R_k$ crosses $\tilde{R_j}$ only if it crosses $R_j$, too. Indeed, every new intersection point of $\partial \tilde{R_j}$ with $\partial R_k$ lies on the segment $[x^j_1,x^j_2] \times (y^i_2 - \varepsilon)$. If $\{R_k,\tilde{R_j}\}$ is a cross, there are two further intersection points of $\partial R_k$ with $\partial \tilde{R_j}$ on the segment $[x^j_1,x^j_2] \times y^j_1$. By the choice of $\varepsilon$ and to position of $R_i$ the pair $\{R_i,R_k\}$ would be a cross, too, but with $R_i \cap R_k \subset R_i \cap R_j$, which contradicts the inclusion-minimality of $R_i \cap R_j$ and proves our claim.

 We proved that no cross is created by the chopping operation described above, and thus there is no \emph{filled} cross among the rectangles $R_1,\ldots,R_{j-1},\tilde{R_j},R_{j+1},\ldots,R_n$. Since $\{R_i,\tilde{R_j}\}$ is not a cross while $\{R_i,R_j\}$ is one, it follows that the total number of crosses has decreased by at least one. Repeating the procedure at most $\binom{n}{2}$ times finally results in a collection of axis-aligned rectangles, which are subsets of the original rectangles, and contain no cross at all. By Theorem~\ref{thm:pseudo-disks} non-crossing connectors exist for the smaller rectangles, which are good for the original rectangles, too.
\end{proof}

\begin{corollary}\label{cor:rectangles}
 It can be tested in $\mathcal{O}(n^2)$ whether or not a set of $n$ axis-aligned rectangles admits a set of non-crossing connectors.
\end{corollary}
\begin{proof}
 By Theorem~\ref{thm:rectangles} we only have to check for every pair of rectangles whether they form a cross. If so the answer is 'No' and if not the answer is 'Yes'.
\end{proof}

\section{NP-Completeness}\label{sec:NP-complete}
In this section we prove NP-completeness of the non-crossing connectors problem. By Proposition~\ref{prop:in-NP} the problem is in NP. We prove NP-hardness, even if the regions and their point sets are very restricted. Let us remark that most of the technicalities in the upcoming reduction, including the use of zones and segment gadgets, are due to the fact that we use \emph{convex} regions only. Dropping convexity but keeping all the other restrictions allows for a much shorter and less technical proof. However, due to space limitations we present only the more technical reduction with convex sets.

\begin{theorem}\label{thm:NP-hard-convex}
 The non-crossing connectors problem is NP-complete, even if the regions are $4$-in\-ter\-sec\-ting convex polygons with at most $8$ corners and for every $i=1,\ldots,n$ the set $P_i$ consists of only two points on the boundary of $R_i$.
\end{theorem}

We prove Theorem~\ref{thm:NP-hard-convex} by a polynomial reduction from \textsc{planar 3-SAT}. In the \textsc{3-SAT} problem we are given a formula $\psi$ in conjunctive normal form where each clause has at most $3$ literals, i.e., positive or negated variables. The \emph{formula graph $G_\psi$} is the bipartite graph whose vertex set is the union of clauses and variables, and whose edge set $E$ is given by $\{x,c\} \in E$ if and only if variable $x$ appears in clause $c$. The \textsc{planar 3-SAT} problem is the \textsc{3-SAT} problem with the additional requirement that $G_\psi$ is a planar graph. It is known~\cite{Lic82} that \textsc{planar 3-SAT} is NP-complete, even if every variable appears in at most $3$ clauses, i.e., $G_\psi$ has maximum degree $3$~\cite{FelKraMidPfe95}.

\medskip
\noindent
\textbf{Zones.} We start by defining a collection of zones, i.e., polygons in the plane, which will later contain our gadgets. It has been proved several times~\cite{Kan93,DujSudWoo05,GanHuKauKob11} that a planar graph with maximum degree $3$ admits a straight line embedding in which every edge has one of the \emph{basic slopes} $0\degree,30\degree$ or $60\degree$, and such an embedding can be computed in linear time. To be precise, possibly three edges on the outer face will have a bend, i.e., not be straight. We take such an embedding of $G_\psi$ and thicken it so that vertices are represented by disks with diameter $\varepsilon$ and edges are rectangles with one side length $\varepsilon$, for some $\varepsilon > 0$ small enough. Let $T$ be an equilateral triangle, smaller than an $\varepsilon$-disk, whose sides have basic slopes and with a tip pointing up. For every clause $c$ let $T(c)$ be a copy of $T$ centered at the position of vertex $c$ in $G_\psi$. We associate the corners of $T(c)$ with the variables in $c$ in the same clockwise order as the corresponding incident edges at vertex $c$ in $G_\psi$. If $c$ has size~$2$, i.e., consists of two variables, then one corner of $T(c)$ is not associated with a variable.

W.l.o.g. every variable $x$ appears at least once positive and once negated in $\psi$. First, let $x$ be a variable that is contained in only two clauses $c_1$ and $c_2$. We connect the corresponding corners of $T(c_1)$ and $T(c_2)$ by a polyline contained in the thickened edges of $G_\psi$, such that each of the following holds. See Figure~\ref{fig:polyline} for an illustration.
\begin{enumerate}[label=\textbf{\alph*)}]
 \item The polyline consists of only constantly many segments, say at most $10$.\label{enum:first}
 \item Every segment has a basic slope, i.e., $0\degree,30\degree$, or $60\degree$.
 \item The angle between any two consecutive segments is $60\degree$.
 \item The first and last segment is attached to the corresponding corner of $T(c_1)$ and $T(c_2)$, respectively, as depicted for the clause $c$ in Figure~\ref{fig:polyline}.\label{enum:last}
\end{enumerate}

\begin{figure}[htb]
 \centering
 \includegraphics{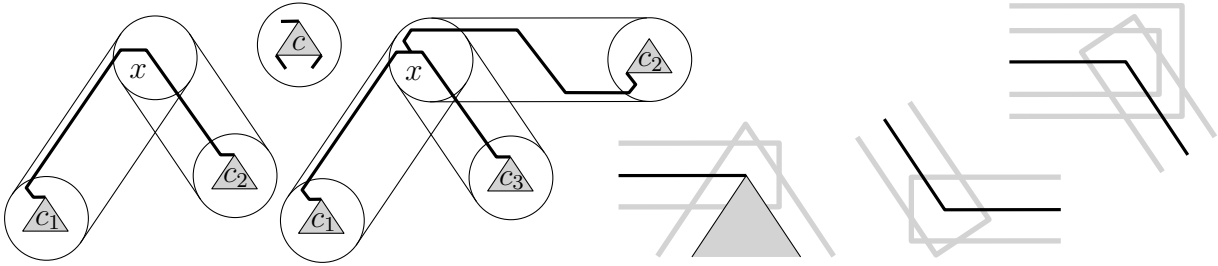}
 \caption{The definition of polylines and zones.}
 \label{fig:polyline}
\end{figure}

\noindent
Let $x$ be a variable that is contained in three clauses $c_1,c_2,c_3$. Let $c_1$ and $c_2$ be the clauses in which $x$ appears with the same sign, i.e., both positive or both negated. We introduce a polyline connecting the corresponding corners of $T(c_1)$ and $T(c_3)$ in the same way as above, i.e., such that \ref{enum:first}--\ref{enum:last} holds. W.l.o.g. assume that one segment $s$ of the polyline is completely contained in the $\varepsilon$-disk for vertex $x$ in $G_\psi$. We introduce a second polyline, which is contained in the thickened edge $\{x,c_2\}$ in $G_\psi$, starts at the endpoint of $s$ that is closer to $c_3$ and ends at the corner of $T(c_2)$ corresponding to $x$. This polyline shall again satisfy \ref{enum:first}--\ref{enum:last}, where in \ref{enum:last} only the last segment is considered, while the first segment shall be one half of segment $s$. See Figure~\ref{fig:polyline} for an illustration.

We denote the two endpoints of every segment by $A$ and $B$, such that whenever two segments $s,s'$ share an endpoint and are not contained in each other, then this endpoint is denoted differently in $s$ and $s'$. Finally we define a zone for every clause and every segment of a polyline. The zone $Z(c)$ of a clause $c$ is an equilateral triangle slightly larger then $T(c)$. The zone $Z(s)$ of a segment $s$ is a thin and long rectangle containing $s$ and with two sides parallel to $s$. How two zones intersect is given in the right of Figure~\ref{fig:polyline}. In the particular case that one segment $s$ is contained in another segment $s'$ we define the zones such that $Z(s)$ is contained in $Z(s')$ as well. This completes the definition of zones. In the remainder of this proof we neither need $\varepsilon$-disks, nor polylines any more. We proceed by defining one gadget, i.e., a set of regions, for every zone.

\medskip
\noindent
\textbf{Clause Gadget.} We define the clause gadget, which consists of $5$ regions as depicted in Figure~\ref{fig:clause-gadget}. For better readability not all regions are drawn convex in the figure. However, for the actual reduction we use the combinatorially equivalent \emph{convex} regions depicted in Figure~\ref{fig:convex-gadgets}.

\begin{figure}[htb]
 \centering
 \includegraphics{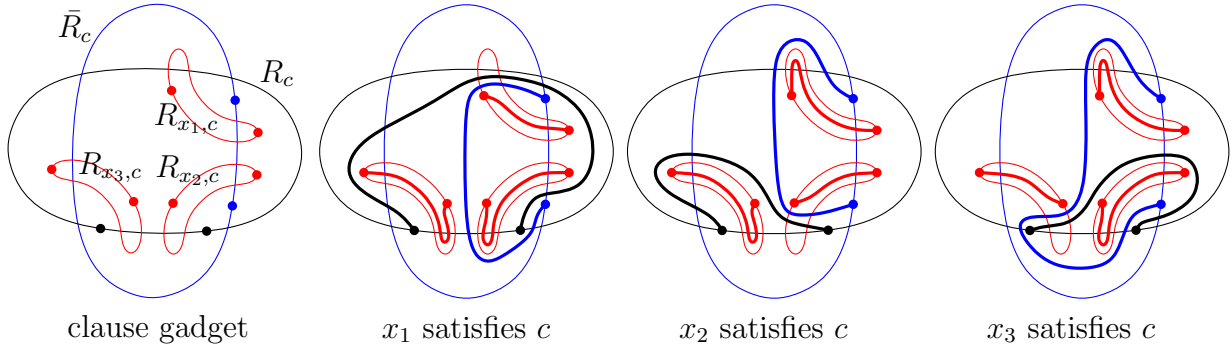}
 \caption{The clause gadget.}
 \label{fig:clause-gadget}
\end{figure}

For every clause $c$ we define a black region $R_{black}(c)$ and a blue region $R_{blue}(c)$, which are $4$-intersecting. The colors are added just for better readability of the figures. We color the $2$-element point set corresponding to every region in the same color as the region. Assume $c$ has size $3$. We define $3$ pairwise disjoint red regions $R_{red}(x,c)$, one for each variable $x$ in $c$, such that the regions appear in the same clockwise order as the edges $\{x,c\}$ around $c$ in $G_\psi$. For every pair $\{x,c\}$ of a variable $x$ and a clause $c$ containing $x$ the red region $R_{red}(x,c)$ has one component inside the black region $R_{black}(c)$, which contains both red points, and one outside $R_{black}(c)$, which does not contain a red point. If the connector $\gamma_{red}(x,c)$ for $R_{red}(x,c)$ is completely contained in $R_{black}(c)$, we say that \emph{variable $x$ satisfies clause $c$}. If the clause has size $2$, then only two of the red regions are associated with the variables. Moreover we put the point of the third red region, which is contained in the blue region $R_{blue}(c)$, anywhere outside the black region $R_{black}(c)$ instead of inside $R_{black}(c)$. Hence this ``artificial variable'' can not satisfy the clause.

It is not difficult to see that the next lemma holds, i.e., Figure~\ref{fig:clause-gadget} verifies the second part of it.

\begin{lemma}\label{lem:satisfy-clause}
 In every set of non-crossing connectors for the $5$ regions of a clause gadget at least one variable satisfies the clause. Moreover, non-crossing connectors do exist as soon as one variable satisfies the clause.
\end{lemma}

\medskip
\noindent
\textbf{Segment Gadget.} For every segment of a polyline we define another gadget, which is very similar to the clause gadget described above and depicted in Figure~\ref{fig:clause-gadget-2}. Actually, we define two possible gadgets, where the second arises from the first by vertically mirroring the left half of it. For every segment $s$ we again define two $4$-intersecting regions, $R_{black}(s)$ is black and $R_{blue}(s)$ is blue. We further define two disjoint red regions $R_{red}(A,s)$ and $R_{red}(B,s)$, associated with the endpoints $A,B$ of the segment $s$. Again, each of $R_{red}(A,s),R_{red}(B,s)$ is divided into the part inside $R_{black}(s)$, which contains both red points, and the part outside $R_{black}(s)$. We say that \emph{endpoint $A$ or $B$ satisfies the segment $s$} if the corresponding connector $\gamma_{red}(A,s)$ or $\gamma_{red}(B,s)$ is completely contained in $R_{black}(s)$. We deduce a statement similar to Lemma~\ref{lem:satisfy-clause}.

\begin{figure}[htb]
 \centering
 \includegraphics{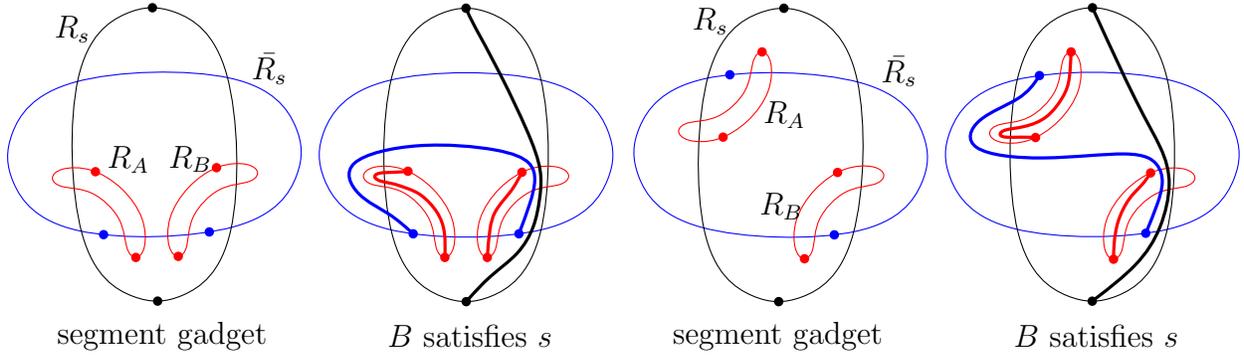}
 \caption{The two possible segment gadgets.}
 \label{fig:clause-gadget-2}
\end{figure}

\begin{lemma}\label{lem:satisfy-segment}
 In every set of non-crossing connectors for the $4$ regions of a segment gadget at least one endpoint satisfies the segment. Moreover, non-crossing connectors do exist as soon as one endpoint satisfies the segment.
\end{lemma}

\medskip
\noindent
\textbf{Putting things together.} We finally describe the regions corresponding to the formula $\psi$. For every clause $c$ we define a clause gadget as depicted in the left of Figure~\ref{fig:convex-gadgets}, which is completely contained in the zone $Z(c)$ corresponding to $c$. Similarly, for every segment $s$ we define a segment gadget as depicted in the right of Figure~\ref{fig:convex-gadgets}, which is completely contained in the zone $Z(s)$ corresponding to $s$. The middle part of a segment gadget, which is highlighted in Figure~\ref{fig:convex-gadgets}, is stretched such that the gadget reaches from one end of the zone to the other. Note that both gadgets are combinatorially equivalent to the ones in Figure~\ref{fig:clause-gadget} and Figure~\ref{fig:clause-gadget-2}, and consist of solely convex polygons with at most $8$ corners. 

\begin{figure}[htb]
 \centering
 \includegraphics[width=0.9\textwidth]{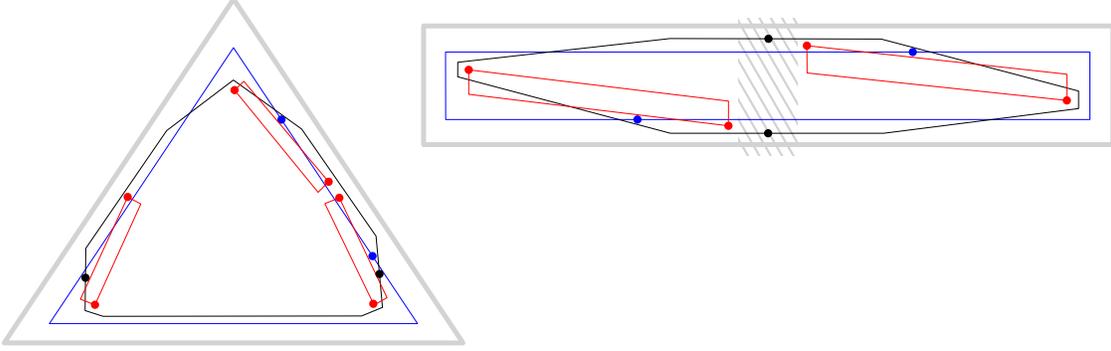}
 \caption{Clause gadget and segment gadget with solely convex polygons with at most $8$ corners.}
 \label{fig:convex-gadgets}
\end{figure}

We let two gadgets intersect as specified in Figure~\ref{fig:2-by-2}. Note that which segment gadget we use for a segment $s$ depends on where the angle of $60\degree$ at either end of the zone $Z(s)$ lies. Consider two intersecting gadgets with intersecting red regions $R_{red}$ and $R'_{red}$, and $4$-intersecting black regions $R_{black}$ and $R'_{black}$. Let $\gamma_{red}$ and $\gamma'_{red}$ be the connector for $R_{red}$ and $R'_{red}$, respectively. It is not difficult to verify that if $\gamma_{red}$ is contained in $R_{black}$, then $\gamma'_{red}$ is not contained in $R'_{black}$. And similarly if $\gamma'_{red}$ is contained in $R'_{black}$, then $\gamma_{red}$ is not contained in $R_{black}$. In other words, only of the two segments/clauses can be satisfied by the endpoint/variable corresponding to the intersection. More formally, we have proven the following lemma.

\begin{figure}[htb]
 \centering
 \includegraphics[width=0.7\textwidth]{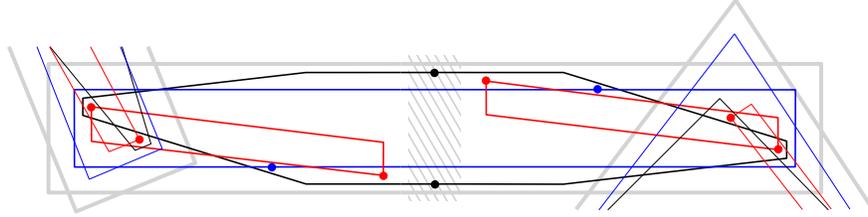}
 \caption{Overlap of a segment gadget with another segment gadget and a clause gadget.}
 \label{fig:2-by-2}
\end{figure}

\begin{lemma}\label{lem:segment-intersection}
 If $R_{red}(A,s) \cap R_{red}(B,s') \neq \emptyset$ for two segments $s,s'$, then $A$ does not satisfy $s$ or $B$ does not satisfy $s'$. Similarly, if $R_{red}(A,s) \cap R_{red}(x,c)$ for some variable $x$ in a clause $c$, then $A$ does not satisfy $s$ or $x$ does not satisfy $c$.
\end{lemma}

Figure~\ref{fig:3-by-2} depicts the mutual overlapping of three segment gadgets. The next lemma can be verified by carefully investigating the intersection pattern of red and black regions.

\begin{figure}[htb]
 \centering
 \includegraphics[width=0.5\textwidth]{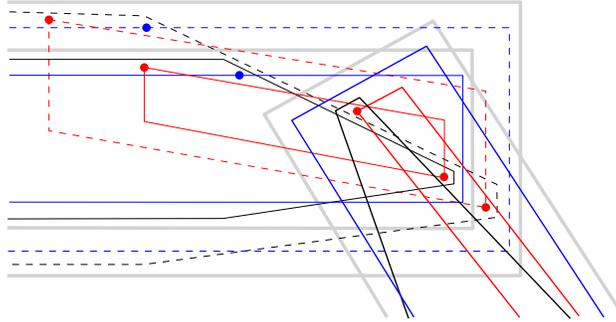}
 \caption{Three mutually overlapping segment gadgets.}
 \label{fig:3-by-2}
\end{figure}

\begin{lemma}\label{lem:segment-3-intersection}
 Suppose $R_{red}(A_1,s_1) \cap R_{red}(A_2,s_2) \cap R_{red}(B_3,s_3) \neq \emptyset$ for three segments $s_1,s_2,s_3$. If $A_1$ satisfies $s_1$ or $A_2$ satisfies $s_2$, then $B_3$ does not satisfy $s_3$.
\end{lemma}

We are now ready to prove Theorem~\ref{thm:NP-hard-convex}.

\noindent
\textbf{Proof of Theorem~\ref{thm:NP-hard-convex}:}
 Let $\psi$ be a 3-SAT formula, and $G_\psi$ be planar with maximum degree $3$. We define an instance $\mathcal{I}$ of the non-crossing connectors problem as described above, consisting of linearly many regions, which are convex polygons with at most $8$ corners and constantly many different slopes. We claim that $\psi$ is satisfiable if and only if non-crossing connectors exist for $\mathcal{I}$.

 First consider a set of non-crossing connectors for $\mathcal{I}$. By Lemma~\ref{lem:satisfy-clause} every clause is satisfied by at least one variable. If variable $x$ satisfies clause $c$, we set $x$ to true if $x$ is positive in $c$ and false if $x$ is negated in $c$. In case some variable has not received a truth value, we choose one arbitrarily. If this truth assignment is consistent, i.e., no variable is set to true and false at the same time, then it clearly satisfies formula $\psi$. So consider a variable $x$ receiving a truth assignment from clause $c$, i.e., $x$ satisfies $c$. We show that $x$ does not satisfy a clause $c'$, in which $x$ appears with the opposite sign. Consider the chain of segments $s_1,\ldots,s_k$ such that $R_{red}(A_1,s_1) \cap R_{red}(x,c) \neq \emptyset$, $R_{red}(B_k,s_k) \cap R_{red}(x,c') \neq \emptyset$, and $R_{red}(B_i,s_i) \cap R_{red}(A_{i+1},s_{i+1}) \neq \emptyset$, for every $i=1,\ldots,k-1$. By Lemma~\ref{lem:segment-intersection} $A_1$ does not satisfy $s_1$. Then by Lemma~\ref{lem:satisfy-segment} $B_1$ satisfies $s_1$. Then again by Lemma~\ref{lem:segment-intersection}, or may be Lemma~\ref{lem:segment-3-intersection}, $A_2$ does not satisfy $s_2$. and by Lemma~\ref{lem:satisfy-segment} $B_2$ satisfies $s_2$. Iterating this pattern yields that $B_k$ satisfies $s_k$ and thus $x$ does not satisfy $c'$, which is what we wanted to prove.

 Secondly, we consider a satisfying truth assignment for formula $\psi$ and want to conclude that there is a collection of non-crossing connectors for $\mathcal{I}$. We define the connectors for the clause gadget similarly to Figure~\ref{fig:clause-gadget}, such that $x$ satisfies $c$ if and only if the variable $x$ is assigned true and appears positive in $c$, or is assigned false and appears negated in $c$. Since $\psi$ is satisfied, every clause $c$ has at least one such variable $x$, so by Lemma~\ref{lem:satisfy-clause} such non-crossing connectors do exist. Then we define the remaining connectors along the chain of segments starting at a clause $c$ that is satisfied by $x$ and ending at a clause $c'$ that is not satisfied by $x$. By a reasoning similar to the previous one, we can construct non-crossing connectors this way for the entire instance $\mathcal{I}$.
\hfill $\square$

\section{Conclusions}\label{sec:conclusions}
In this paper we investigated the computational complexity of the non-crossing connectors problem, i.e., given pairwise disjoint finite point sets $P_1,\ldots,P_n$ and simply connected regions $R_i \supset P_i$ for every $i=1,\ldots,n$, is there a set of pairwise disjoint curves $\gamma_1,\ldots,\gamma_n$, called connectors, such that $P_i \subset \gamma_i \subset R_i$ for every $i=1,\ldots,n$. We proved that the existence of non-crossing connectors can be tested in polynomial time if the regions are pseudo-disks or axis-aligned rectangles. It might be worthwhile to derive from our proofs polynomial-time algorithms to actually compute non-crossing connectors. We proved that the problem is NP-complete for $4$-intersecting convex regions, even if every $P_i$ consists of $2$ elements. However, we do not know the complexity in case each $R_i$ is the convex hull of the corresponding $P_i$. Moreover, it is interesting to consider other sets of regions, like ellipsoids or isosceles triangles with horizontal bases.

Instead of fixing the position of the points in $P_i$ one can consider a set of possible positions for every such point. From previous results~\cite{many10,many11} it follows that this variant is NP-complete if every point has at most $3$ possible positions, but this proof again does not work if $R_i = \operatorname{conv}(P_i)$ for $i=1,\ldots,n$. What if we want to connect pairs of vertical straight segments by non-crossing curves, each within the convex hull of the corresponding segments?

Furthermore, we could allow $P_i \cap P_j \neq \emptyset$ for $i \neq j$ and allow connectors $\gamma_i,\gamma_j$ to intersect in $P_i \cap P_j$. If $|P_i| = 2$ for every $i$, this corresponds to drawing a given planar graph with fixed vertex positions and curved edges, each lying within a prescribed region. In this variant non-crossing connectors sometimes do not exist even when regions are pseudo-disks.

\section*{Acknowledgements}

We would like to thank Maria Saumell, Stefan Felsner and Irina Mustata for fruitful discussions.


\bibliography{lit}
\bibliographystyle{amsplain}

\end{document}